\documentclass[11pt]{article}

\usepackage{amsmath}
\usepackage{amssymb}
\usepackage{fancyhdr}
\usepackage{amsthm}
\usepackage{moreverb}
\usepackage{algorithmic}
\usepackage{algorithm}
\usepackage{rotating}
\usepackage{endnotes}
\newcommand{\n}{\tilde{n}}

\newtheorem{theorem}{Theorem}
\newtheorem{lemma}[theorem]{Lemma}

\newtheorem{cor}[theorem]{Corollary}

\theoremstyle{definition}

\begin{document}

\medskip                        

\thispagestyle{plain}
\title{\Large Tight Lower Bounds on Envy-Free Makespan Approximation}

\author{Amos Fiat\footnote{Tel Aviv University, Tel Aviv, Israel.} \ and Ariel Levavi\footnote{University of California San Diego, San Diego, CA.}}
\date{}

\maketitle
\begin{abstract}
In this work we give a tight lower bound on makespan approximations for envy-free allocation mechanism dedicated to scheduling tasks on unrelated machines. Specifically, we show that no mechanism exists that can guarantee an envy-free allocation of jobs to $m$ machines with a makespan of less than a factor of $O(\log m)$ of the minimal makespan. Combined with previous results, this paper definitively proves that the optimal algorithm for obtaining a minimal makespan for any envy-free division can at best approximate the makespan to a factor of $O(\log m)$.
\end{abstract}

\section{Introduction}

Consider the scenario in which there is a set of tasks and a workforce that is commissioned to complete it. The tasks we are interested in are {\it indivisible}, that is, we can assign more than one job to each worker but two workers cannot both work on the same task. One goal is to complete all the tasks in the shortest period of time. However, each worker is specialized in his own way and ranks the difficulty of performing each task differently from his colleagues. We would also like to allocate the tasks in such a way that no worker would prefer to complete the workload of a colleague over his own. This problem is the focus of this paper.

Determining fair division is at the heart of a large body of research in computer science. One of its earliest occurrences in literature was in 1947 when Neyman, Steinhaus, Banach and Bronislaw modeled it as the problem of how to find a fair partitioning of cake (\cite{steinhaus1}, \cite{steinhaus2}). Since then, several books including \cite{barbanel}, \cite{bt1}, \cite{bt2}, \cite{meertens}, \cite{hougaard}, \cite{moulin} have been dedicated to the subject. The problem is generally described as a way of assigning $n$ jobs to be processed by $m$ machines or agents in a fair manner. One of the reasons that this area of research is so rich is that there are multiple ways to characterize a fair allocation. One way to do so is to consider the division that preserves {\it envy-freeness}, the notion that no agent would be better off if he were assigned the set of jobs given to another (\cite{ds}, \cite{foley}). In the scenario where jobs can be divided among more than one machine, one solution would be to divide all jobs equally among all agents (although depending on the set-up of the problem this division might be ill-defined, i.e. if one agent takes an infinitely long time to process a specific job).

Determining a fair allocation when jobs are not divisible is less straightforward. In order to furnish a solution we must define a {\it mechanism} that determines an allocation as well as payments either to or from the agents, or between agents and the mechanism or agents among themselves. We consider the utility of each agent to be quasi-linear, i.e. the difference between the payment he receives and the cost to process his assignment of jobs. 

When determining an optimal envy-free solution other goals can be considered such as revenue optimization or economic efficiency. The additional goal we described in our earlier example was {\it makespan} minimization, or the intention to minimize the longest processing time of jobs on any one machine. In their paper \cite{himsz}, Hartline et. al. considers a task schedule for $m$ machines in which the minimum makespan for any indivisible allocation is 1. They then go on to show that no mechanism exists that can provide an envy-free indivisible allocation of with a makespan of less than $2-1/m$. In addition they provide an algorithm to find an envy-free indivisible allocation that upper bounds the makespan by $(m+1)/2$. Two years later, Cohen et. al. generalized and tightened the bounds on makespan approximation \cite{cffko}. In there paper, they show that there does not exist a mechanism that provides on envy-free division with a makespan of less than $O(\log m/\log\log m)$ times the optimal, and produce a polynomial time algorithm that finds an envy-free scheduling that approximates the minimal makespan by a factor of $O(\log m)$. 

Our contribution is to tighten the lower bound on makespan approximation to the upper bound. Specifically, we show that no mechanism exists that can guarantee an envy-free allocation of jobs to $m$ machines with a makespan of less than a factor of $O(\log m)$ of the minimal makespan. This result definitively proves that the optimal algorithm for obtaining a minimal makespan for any envy-free division can at best approximate the makespan to a factor of $O(\log m)$.

\section{Preliminaries}

The scheduling problem that we are interested in is the following: We have $n$ indivisible jobs and $m$ machines. We are given a {\it cost matrix} $(c_{i,j})_{1\leq i\leq m,1\leq j\leq n}$ where $c_{i,j}$ is the cost of running job $j$ on machine $i$. The {\it allocation matrix} $(a_{i,j})_{1\leq i\leq m,1\leq j\leq n }$ specifies which jobs are assigned to run on which machines, so that $a_{i,j} = 1$ if we run job $j$ on the $i$th machine and $a_{i,j} = 0$ otherwise. Since our focus is on indivisible jobs, if $a_{i,j} = 1$, then $a_{i',j} = 0$ for all $i'\neq i$. In the case where jobs are divisible $a_{i,j}\in [0,1]$. In both the divisible and indivisible job cases, $\sum_{i=1}^ma_{i,j} = 1$, i.e. we always find an allocation of jobs to machines where every job is processed in its entirety.

Let $\bar{c}_i = (c_{i1},\dots,c_{in})$ denote the $i$th row of the cost matrix $(c_{i,j})$ and let  $\bar{a}_i = (a_{i1},\dots,a_{in})$ denote the $i$th row of the allocation matrix $(a_{i,j})$. Then the {\it load} on machine $i$, or the cost running the jobs assigned to machine $i$ is $\bar{c}_i\cdot\bar{a}_i = \sum_{j=1}^nc_{i,j}a_{i,j}$. The {\it makespan} of an assignment is the maximum load on any machine, or $\max_{1\leq i\leq m}\bar{c}_i\cdot\bar{a}_i$.

We can formulate the problem of finding the minimum makespan for indivisible jobs as an integer programming problem and for divisible jobs as a linear programming problem. In 1990, Lenstra, et. al. introduced a 2-approximation algorithm for finding the minimum makespan for indivisible jobs, and showed that there does not exist a $\rho$-approximation algorithm for finding the minimum makespan for $\rho<3/2$ unless $P=NP$ \cite{lst}.

In this formulation we consider each of the $m$ machines as a selfish agent. An allocation function $a$ is mapped to the $m\times n$ cost matrix $c$ so that $a(c) = (a_{i,j})$. As before let $\bar{c}_i = (c_{i1},\dots,c_{in})$ and  $\bar{a}(c)_i = (a(c)_{i1},\dots,a(c)_{in})$ denote the $i$th row of $(c_{i,j})$ and $a(c)$, respectively. Let $p$ denote a payment function that is a mapping from $c$ to $\Bbb R^m$, and let $p(c)_i$ denote the $i$th coordinate of $p(c)$.

We define a {\it mechanism} as a pair of functions, $M=<a,p>$ where $a$ is the allocation function and $p$ is the payment function. Given a mechanism $<a,p>$ with a cost function $(c_{i,j})$, the {\it utility} of agent $i$ is $p(c)_i-\bar{c}_i\cdot\bar{a}_i$. A mechanism is considered {\it envy-free} if no agent can increase his utility by trading his job allocation and payment with another player. More formally, a mechanism is envy-free if, for all $j\in 1..n$,
$$
p(c)_i-\bar{c}_i\cdot\bar{a}_i\geq p(c)_j-\bar{c}_i\cdot\bar{a}_j.
$$
We call an allocation function {\it envy-free implementable} ({\it EF}-implementable) if there exists a payment function $p$ such that mechanism $<a,p>$ is envy-free.

An allocation function is {\it locally-efficient} if for all cost matrices $c$ and permutations $\pi$ of $1,\dots,m$, we have
$$
\sum_{i=1}^m\bar{c}_i\cdot\bar{a}_i\leq \sum_{i=1}^m\bar{c}_i\cdot\bar{a}_{\pi(i)}.
$$
Hartline, et. al. introduced the following useful theorem \cite{himsz}.

\begin{theorem}
An allocation is {\it EF}-implementable if and only if it is locally-efficient.
\end{theorem}

\section{Main Result: Lower Bound on Envy-Free Makespan Approximation}

We give a lower bound of $\Omega(\log m)$ on the makespan achievable by any envy-free allocation of jobs. 

Let $n=\frac{\n}{\log \n}+1$ be the number of jobs for some $\n\in\Bbb Z^+$. The number of machines is $m=n+l$ where $2^l=\log \n$. Let $c$ denote a cost matrix where $c_{i,j}$ gives the cost of running job $j$ on machine $i$. For this cost matrix, we have

$$c=\left(\begin{array}{llllll}
1& \infty & \infty &\dots&\infty&\infty \\
1-\frac{\log \n}{2\n} & 1&\infty&\dots&\infty&\infty\\
1-\frac{2\log\n}{2\n} & 1-\frac{\log \n}{2(\n-1)}&1&\dots&\infty&\infty\\
1-\frac{3\log \n}{2\n} & 1-\frac{2\log \n}{2(\n-1)}&1-\frac{\log \n}{2(\n-2)}&\dots&\infty&\infty\\
\vdots &\vdots &\vdots &\vdots &\vdots  &\vdots \\
1- \left(\frac{\n}{\log \n}-1	\right)\frac{\log \n}{2\n} & 1- \left(\frac{\n}{\log \n}-2	\right)\frac{\log \n}{2(\n-1)}&1- \left(\frac{\n}{\log \n}-3	\right)\frac{\log \n}{2(\n-2)} &\dots &1&\infty\\
1/2 & 1- \left(\frac{\n}{\log \n}-1	\right)\frac{\log \n}{2(\n-1)}&1- \left(\frac{\n}{\log \n}-2	\right)\frac{\log \n}{2(\n-2)} &\dots& > 1/2&1\\
\hline\\
2&2&2&\dots&2&2\\
4&4&4&\dots&4&4\\
\vdots&\vdots&\vdots&\vdots&\vdots&\vdots\\
2^l&2^l&2^l&\dots&2^l&2^l
\end{array}\right)$$

Each row $i$ for $1\leq i\leq n+l$ gives the costs for the $i$th machine and each entry $c_{i,j}$ in the matrix denotes the cost of running job $j$ on machine $i$. The horizontal line lies between machines $n$ and $n+1$. For $1\leq i\leq n$, the cost of running job $j$ on machine $i$ is given by
$$
c_{i,j}=\left\{\begin{array}{lll}
1 & \mbox{if }i=j\\
1-\frac{(i-j)\log\n}{2(\n-j+1)} & \mbox{if }i>j \\
\infty &\mbox{if }i<j
\end{array}\right..
$$

Note that for $1\leq i<n$ and $i>j$, we have $c_{i,j}-c_{i+1,j}=\frac{\log \n}{2(\n-j+1)}$. For $n< i\leq n+ l$, the cost to process any job on machine $ i$ is $2^{i-n}$.

\begin{lemma}\label{lem:cost1/2}
For $1\leq i\leq n+l$ and $1\leq j\leq n$, we have $c_{i,j}> 1/2.$
\end{lemma}

\begin{proof}
The lemma statement is trivially true for $i\leq j$ or $n< i\leq n+l$. Therefore, we need only consider the cases where $j<i<n$. Consider the function $f(x,y) = 1-\frac{(x-y)\log\n}{2(\n-y+1)}$ where $1\leq x\leq \n/\log\n$ and $1\leq y\leq x-1$. Since the function is linear, it's critical points will be where the boundaries of the domain intersect. The three boundaries are at $y=1$, $x=\n/\log \n$, and $y=x-1$, and so the points of intersection are $(2,1)$, $(\n/\log\n,1)$, and $(\n/\log\n,\n/\log\n-1)$. We have $f(2,1) = 1-\log \n /(2\n)$, $f(\n/\log\n,1) = 1 - (\n/\log\n)\log\n/(2\n)$ and $f(\n/\log\n,\n/\log\n-1)= 1-\log\n/(2(\n-\n/\log\n))$ all of which are in the range of $(1/2,1)$.
 \end{proof}

For this cost matrix, the optimal makespan is 1. We reach this makespan when $i=j$. Since the cost of running any job on any machine is strictly greater than $1/2$, if more than one job is run per machine the makespan will be more than 1. Any other permutation of jobs would require at least one job $j$ to be run on some machine $i$ for $i<j$ or for $n<i\leq n+ l$. Either of these scenarios would give us a makespan of at least 2.

We show that any envy-free makespan for this cost matrix is lower bounded by $\log n$. More specifically, we show that no matter how we partition the $n$ jobs into $n+l$ bundles, any locally-efficient assignment of the bundles has a makespan of at least $\log n$. This establishes that there does not exist an algorithm that can always find a makespan of less than $\log n$.

\begin{theorem}\label{thm:makespan}
For any partition of $n$ jobs into bundles, the makespan  for every locally efficient assignment of  bundles is at least $2^l=\log \n$.
\end{theorem}

Before we prove this theorem, we introduce the following useful lemma.

\begin{lemma}\label{lem:2^l}
For the cost matrix $(c_{i,j})$ defined, and makespan of less than $2^l=\log\n$ has the following properties:
\begin{enumerate}
\item Fewer than $2^{l+1}$ jobs run on each machine. 
\item Fewer than $2^l/2^{i-n}$ jobs run on each machine $n+i$ for $n<i\leq n+l$.
\item The total number of jobs running on machines $n+1,\dots,n+l$ is fewer than $2^l$.
\end{enumerate}
\end{lemma}

\begin{proof}
Property $(1)$ follows directly from Lemma \ref{lem:cost1/2}; Property $(2)$ holds since $c_{i,j} = 2^{i-n}$ for $n<i\leq n+l$; and $(3)$ follows from $(2)$ because $\sum_{i=n+1}^{n+1}c_{i,j} < \sum_{i=n+1}^{n+1}2^l/2^{i-n}=2^l$.
\end{proof}

\begin{proof}[Proof of Theorem \ref{thm:makespan}]
Consider an arbitrary partition of the $n$ jobs into $n+l$ bundles with a makespan of less than $2^l$. Suppose that this assignment is locally-efficient. In order to prove this theorem by contradiction, we must provide a permutation of the assignment that decreases that total cost over all jobs. Since the cost of running a job on machine $n+l$ is $2^l$, there are no jobs assigned to run on machine $n+l$. Therefore, the permutation we will provide is the one in which each bundle of jobs assigned to machine $i$ is moved to machine $i+1$.

By Lemma \ref{lem:2^l} (1), less than $2^{l+1}$ jobs run on machine $n$, so the increase of cost from moving the bundle of jobs from machine $n$ to machine $n+1$ is less than $2^{l+1}(2-1/2) = 3\cdot 2^l$. For $n<i< n+l$, we have $c_{i+1,j}=2c_{i,j}$, and so moving each bundle from machine $i$ to $i+1$ in this range increases the cost by a factor of 2. By Lemma \ref{lem:2^l} (3), fewer than $2^l$ jobs run on this set of $l$ machines, and so moving each bundle to the next machine would increase the total cost by less than $l\cdot2^l$. Therefore, moving each bundle assigned to machine $i$ to machine $i+1$ for $n\leq i<n+l$ increases the cost of the assignment by less than $(l+3)2^l= (\log\log \n+3)\log \n$. 

By Lemma \ref{lem:2^l} (3), there are fewer than $2^l$ jobs running on machines $n+1,\dots,n+l$, which implies that the total number of jobs running on machines $1,\dots,n$ is greater than $n-2^l$. Pairing this with Lemma \ref{lem:2^l} (1), we know that the total number of jobs running on machines $1,\dots, n-1$ is greater than $n-2^l-2^{l+1}=n-3\cdot 2^l$. As noted earlier, moving any job $j$ from machine $i$ to machine $i+1$ in this range decrease the cost of the job by $\frac{\log \n}{2(\n-j+1)}$. Therefore, the total cost from moving all the bundles on machines $1,\dots,n-1$ decreases by at least $(\frac{\log \n}{2})(H_{\n/\log\n}-H_{3\cdot 2^l})\approx (\frac{\log \n}{2})(\ln \n-\ln\log \n- \ln (3\log \n))$, where $H_k$ is the $k$th harmonic number. 

The decrease of cost from the first $n-1$ machines is larger than the increase of cost from the last $l+1$ machines and so the makespan for any locally efficient assignment must be greater than $2^l$.
\end{proof}

\begin{cor}
For any partition of $n$ jobs into bundles, the makespan  for every envy-free assignment of  bundles is $\Omega(\log m)$.
\end{cor}

\begin{proof}
By Theorem \ref{thm:makespan}, every locally efficient partition has a makespan of at least $\log \n = \log((n-1)\log\n)\geq \log n$. Since $m=n+l=O(n)$ for the cost matrix defined, it holds that it is an $\Omega(\log m)$ approximation.
\end{proof}

\end{document}